\documentclass[a4paper,UKenglish,cleveref, autoref, thm-restate]{lipics-v2021}

\usepackage{amsmath}
\usepackage{amsfonts}
\usepackage{amssymb}
\usepackage{amsmath}
\usepackage{algpseudocode}
\usepackage{tikz}
\usepackage{xspace}
\usepackage[linesnumbered, vlined, ruled]{algorithm2e}
\usepackage{stmaryrd}
\usepackage{float}
\usepackage{enumitem}
\usepackage{trimclip}
\usepackage{mathtools}

\newtheorem{thm}{Theorem}
\newtheorem{lem}[thm]{Lemma}

\title{Formal Primal-Dual Algorithm Analysis}

\titlerunning{Formal Primal-Dual Algorithm Analysis}

\author{Mohammad Abdulaziz}{King's College London, United Kingdom  \and \url{https://mabdula.github.io/}}{mohammad.abdulaziz@kcl.ac.uk}{https://orcid.org/0000-0002-8244-518X}{}
\author{Thomas Ammer}{King's College London, United Kingdom \and \url{https://toamme.github.io/}}{thomas.ammer@kcl.ac.uk}{https://orcid.org/0009-0001-5301-4620}{}
\author{Christoph Madlener}{ }{christoph.madlener96@gmail.com}{https://orcid.org/0000-0002-9577-0061}{}

\authorrunning{M. Abdulaziz, T. Ammer and C. Madlener}

\Copyright{Mohammad Abdulaziz and Thomas Ammer}



\nolinenumbers 
\EventNoEds{2}
\EventShortTitle{CVIT 2016}
\EventAcronym{CVIT}
\EventYear{2016}
\EventDate{December 24--27, 2016}
\EventLocation{Little Whinging, United Kingdom}
\EventLogo{}
\SeriesVolume{42}
\ArticleNo{23}

\lstdefinestyle{inline}{%
    basicstyle=\ttfamily%
}

\lstdefinelanguage{Isabelle}{ morekeywords={for,context,inductive,begin,end,locale,fixes,record,type_synonym,definition,fun,function,primrec,where,lemma,theorem,unfolding,by,
shows,assumes,and,datatype,using,abbreviation,moreover,have,hence,thus,qed,proof,let,ultimately,show,next,in,if,value, thm,corollary,else,then}
    , sensitive=true
    , showstringspaces=true
    , framerule=0pt
    , xleftmargin=2em
    , numbers=left
    , numberstyle=\ttfamily\small
    , firstnumber=1
    , stepnumber=2
    , basicstyle=\ttfamily\small
    , backgroundcolor = \color{white}
    , keywordstyle = {\color{blue}}
    , breaklines=true
    , showspaces=false
    , morecomment=[l]{--}
    , morecomment=[s]{(*}{*)}
    , commentstyle=\color{gray}
    , morestring=[b]"
    , literate={↦}{{$\mapsto$}}{1}
               {∧}{{$\wedge$}}{1}
               {×}{{$\times$}}{1}
               {≡}{{$\equiv$}}{1}
               {∀}{{$\forall$}}{1}
               {∃}{{$\exists$}}{1}  {\\<times>}{{$\times$}}{1}
               {\\<and>}{{$\land$}}{1} {\\<inter>}{{$\cap$}}{1}
        {∈}{{$\in$}}{1} {⇒}{{$\Rightarrow$}}{1} {\\<Rightarrow>}{{$\Rightarrow$}}{1} {\\<And>}{{$\bigwedge$}}{1} {\\<forall>}{{$\forall$}}{1}   {\\<in>}{{$\in$}}{1} {\\<exists>}{{$\exists$}}{1}
        {\\<subseteq>}{{$\subseteq$}}{1} {\\<longrightarrow>}{{$\longrightarrow$}}{1} {\\<Longrightarrow>}{{$\Longrightarrow$}}{1}
        {\\<notin>}{{$\not\in$}}{1}
        {λ}{{$\lambda$}}{1} {\\<lambda>}{{$\lambda$}}{1} {::}{{$::$}}{1} {\\<epsilon>}{{$\epsilon$}}{1} {\\<alpha>}{{$\alpha$}}{1}
        {\⊆}{{$\subseteq$}}{1} {\\<subset>}{{$\subset$}}{1} {\\<^sub>m}{{$_m$}}{1} {\\<longleftrightarrow>}{{$\longleftrightarrow$}}{3}
        {\\<pi>}{{$\pi$}}{1} {\\<delta>}{{$\delta$}}{1} {⟦}{{$\llbracket$}}{1} {⟧}{{$\rrbracket$}}{1} {\\<lbrakk>}{{$\llbracket$}}{1} {\\<rbrakk>}{{$\rrbracket$}}{1}
        {⟹}{{$\Longrightarrow$}}{3} {\\<not>}{{$\lnot$}}{1} {\\<le>}{{$\le$}}{1} {\\<rightharpoonup>}{{$\rightharpoonup$}}{2}
        {\\<^sub>\\<V>}{{$_{\mathcal V}$}}{1} {\\<lparr>}{{$\llparenthesis$}}{1} {\\<rparr>}{{$\rrparenthesis$}}{1}
        {\\<leftarrow>}{{$\leftarrow$}}{1} {\\<^sub>\\<O>}{{$_{\mathcal O}$}}{1} {\\<^sub>I}{{$_{\texttt{I}}$}}{1}
        {\\<^sub>G}{{$_{\texttt{G}}$}}{2} {\\<phi>}{{$\varphi$}}{1} {\\<Phi>}{{$\Phi$}}{1} {\\<psi>}{{$\psi$}}{1} {\\<Psi>}{{$\Psi$}}{1}
        {\\<^sub>S}{{$_{\texttt S}$}}{1} {\\<inverse>}{{$^{-1}$}}{1} {\\<^sub>O}{{$_{\texttt O}$}}{1} {\\<^bold>\⋀}{{$\bm\bigwedge$}}{1} {⋀}{{$\bigwedge$}}{1}
        {\\<^bold>\\<or>}{{$\bm\lor$}}{1} {\\<^sub>G}{{$_{\texttt G}$}}{1} {\\<Pi>}{{$\Pi$}}{1} {\\<^sub>I}{{$_{\texttt I}$}}{1} {\≠}{{$\neq$}}{1}
        {\\<bottom>}{{$\bot$}}{1} {\\<^sub>+}{{$_\texttt +$}}{1} {\\<^bold>\\<and>}{{$\bm\land$}}{1} {\\<^bold>\\<not>}{{$\bm\lnot$}}{1}
        {\\<^sub>1}{{$_1$}}{1} {\\<subpset>1}{{$\supset$}}{1} {\\<^sub>2}{{$_2$}}{1} {\\<A>}{{$\mathcal A$}}{1} {\\<Turnstile>}{{$\models$}}{2} {\\<^sub>∀}{{$_\forall$}}{1} {\\<noteq>}{{$\neq$}}{1} 
        {\\<^sub>0}{{$_0$}}{1} {\\<tau>}{{$\tau$}}{1}  {\\<^sub>\\<Omega>}{{$_\Omega$}}{1} {\\<^sub>V}{{$_V$}}{1} {\\<^bold>\\<Or>}{{$\bm\bigvee$}}{1}
        {\\<^sub>P}{{$_\texttt P$}}{1} {\\<^sub>X}{{$_\texttt X$}}{1} {⟹}{{$\Longrightarrow$}}{2} {\\<or>}{{$\lor$}}{1} {\\<^sub>\\<pi>}{{$_\pi$}}{1}
        {\\<^sub>s}{{$_s$}}{1} {\\<^sub>t}{{$_t$}}{1} {\\<^sub>a}{{$_a$}}{1} {\\<^sub>r}{{$_r$}}{1} {\\<^sub>t}{{$_t$}}{1} {\\<^sub>e}{{$_e$}}{1} {\\<^sub>n}{{$_n$}}{1} {\\<^sub>d}{{$_d$}}{1} {\\<^sub>i}{{$_i$}}{1} {\\<^sub>v}{{$_v$}}{1} {\\<^sub>j}{{$_j$}}{1} {\\<^sub>b}{{$_b$}}{1} {\∩}{{$\cap$}}{1} {\\<union>}{{$\cup$}}{1} {\\<Union>}{{$\bigcup$}}{1} {\\<^sup>c\\<TTurnstile>\\<^sub>=}{{${}^c\models_=$}}{1}
        {\\<open>}{{<}}{1} {\\<close>}{{>}}{1} {\\<langle>}{{$\langle$}}{1} {\\<rangle>}{{$\rangle$}}{1} {\\<ge>}{{$\ge$}}{1} {\\<^sup>-\\<^sup>1\\<^sub>C}{{$^{\texttt{-1}}_\texttt{C}$}}{2} {\\<^sup>+\\<^sub>C}{{$^\texttt{+}_\texttt{C}$}}{1} {\\<circ>\\<^sub>C}{{$\circ^\texttt C$}}{1} {\\<top>\\<^sub>C}{{$\top_\texttt{C}$}}{2} {\\<bottom>\\<^sub>C}{{$\bot_\texttt{C}$}}{2} {\\<not>\\<^sub>C}{{$\neg_\texttt{C}$}}{2}
        {\\<squnion>\\<^sub>C}{{$\sqcup_\texttt{C}$}}{2} {\\<sqinter>\\<^sub>C}{{$\sqcap_\texttt{C}$}}{2} {\∃\\<^sub>C}{{$\exists_\texttt{C}$}}{2}  {∀\\<^sub>C}{{$\forall_\texttt{C}$}}{2} {=\\<^sub>C}{{$=_\texttt{C}$}}{2} {\\<sigma>}{{$\sigma$}}{2} {\∉}{{$\notin$}}{1} {⊕}{{$\oplus$}}{1} {\\<nexists>}{{$\nexists$}}{1} {\\<setminus>}{{$\setminus$}}{1} {\ }{\ }{1} {\\<bullet>}{$\bullet$}{1} {\\<^sup>T}{$^T$}{1} {_}{-}{1}
        {"}{}{1} {\\pi^*}{$\pi^*$}{1} {\\<pi>\\<^sup>*}{$\pi^*$}{1} {\\<V>}{$\mathcal{V}$}{1} {\\<E>}{$\mathcal{E}$}{1} {\\<M>}{$\mathcal{M}$}{1} {\\<F>}{$\mathcal{F}$}{1} {\\<bar>}{$|$}{1} {\\of}{{\color{blue}of}}{1} {\\and}{{\color{blue}and}}{1}
        {\\<nu>}{$\nu$}{1}
}

\begin{document}
\maketitle

\begin{abstract}
We present an ongoing effort to build a framework and a library in Isabelle/HOL for formalising primal-dual arguments for the analysis of algorithms.
We discuss a number of example formalisations from the theory of matching algorithms, covering classical algorithms like the Hungarian Method, widely considered the first primal-dual algorithm, and modern algorithms like the Adwords algorithm, which models the assignment of search queries to advertisers in the context of search engines.
\end{abstract}

\providecommand{\insts}{}
\renewcommand{\insts}{\ensuremath{\Delta}}
\providecommand{\inst}{\ensuremath{\tvsal}}
\newcommand{\act}{\ensuremath{\pi}}
\newcommand{\asarrow}[1]{\vec{#1}}
\renewcommand{\vec}[1]{\overset{\rightarrow}{#1}}
\newcommand{\as}{\ensuremath{\vec{{\act}}}}

\newcommand{\etc}{\textit{etc.}}
\newcommand{\versus}{\textit{vs.}}

\newcommand{\Ie}{I.e.}
\newcommand*{\ie}{i.e.\ }
\newcommand{\abziz}[1]{\textcolor{brown}{#1}}
\newcommand{\sublist}[2]{ \ensuremath{#1} \preceq\!\!\!\raisebox{.4mm}{\ensuremath{\cdot}}\; \ensuremath{#2}}
\newcommand{\subscriptsublist}[2]{\ensuremath{#1}\preceq\!\raisebox{.05mm}{\ensuremath{\cdot}}\ensuremath{#2}}
\newcommand{\PLS}{\Pi^\preceq\!\raisebox{1mm}{\ensuremath{\cdot}}}
\newcommand{\PLScharles}{\Pi^d}
\newcommand{\execname}{\mathsf{ex}}
\newcommand{\IndHyp}{\mathsf{IH}}
\newcommand{\exec}[2]{#2(#1)}

\newcommand{\ancestorssymbol}{\textsf{\upshape ancestors}}
\newcommand{\ancestors}{\ancestorssymbol}
\newcommand{\satpreas}[2]{\ensuremath{sat_precond_as(s, \as)}}
\newcommand{\proj}[2]{\ensuremath{#1{\downharpoonright}_{#2}}}
\newcommand{\dep}[3]{\ensuremath{#2 {\rightarrow} #3}}
\newcommand{\deptc}[3]{\ensuremath{#2 {\rightarrow^+} #3}}
\newcommand{\negdep}[3]{\ensuremath{#2 \not\rightarrow #3}}
\newcommand{\leavessymbol}{\textsf{\upshape leaves}}
\newcommand{\leaves}{\leavessymbol}

\newcommand{\childrensymbol}{\textsf{\upshape children}}
\newcommand{\children}[2]{\mathcal{\childrensymbol}_{#2}(#1)}
\newcommand{\succsymbol}{\textsf{\upshape succ}}
\newcommand{\succstates}[2]{\succsymbol(#1, #2)}
\newcommand{\concat}{\#}
\newcommand{\RG}{\cite{Rintanen:Gretton:2013}\ }
\newcommand{\cupdot}{\charfusion[\mathbin]{\cup}{\cdot}}
\newcommand{\bigcupdot}{\charfusion[\mathop]{\bigcup}{\cdot}}
\newcommand{\cuparrow}{\charfusion[\mathbin]{\cup}{{\raisebox{.5ex} {\smathcalebox{.4}{\ensuremath{\leftarrow}}}}}}
\newcommand{\bigcuparrow}{\charfusion[\mathop]{\bigcup}{\leftarrow}}
\newcommand{\finiteunion}{\cuparrow}
\newcommand{\finitemap}{\ensuremath{\sqsubseteq}}
\newcommand{\dgraph}{dependency graph}
\newcommand{\domain}[1]{{\sc #1}}
\newcommand{\solver}[1]{{\sc #1}}
\providecommand{\problem}[1]{\domain{#1}}
\renewcommand{\v}{\ensuremath{\mathit{v}}}
\providecommand{\vs}[1]{\domain{#1}}
\renewcommand{\vs}{\ensuremath{\mathit{vs}}}
\newcommand{\VS}{\ensuremath{\mathit{VS}}}
\newcommand{\Aut}{\ensuremath{\mathit{Aut}}}
\newcommand{\Inst}[2]{\ensuremath{\mathit{#2 \rightarrow_{#1} #1}}}
\newcommand{\Image}{\ensuremath{\mathit{Im}}}
\newcommand{\Img}[2]{\protect{#1 \llparenthesis #2 \rrparenthesis}}
\newcommand{\SND}{\ensuremath{\mathit{\pi_2}}}
\newcommand{\FST}{\ensuremath{\mathit{\pi_1}}}
\newcommand{\tvsal}{{\pitchfork}}
\newcommand{\nauty}{CGIP}

\newcommand{\pwinter}{\ensuremath{\mathit{\bigcap_{pw}}}}

\newcommand{\dom}{\ensuremath{\mathit{\mathcal{D}}}}
\newcommand{\codom}{\ensuremath{\mathcal{R}}}

\newcommand{\map}{\ensuremath{\mathit{map}}}
\newcommand{\BIJEC}{\ensuremath{\mathit{bij}}}
\newcommand{\INJ}{\ensuremath{\mathit{inj}}}
\newcommand{\funion}{\ensuremath{\overset{\leftarrow}{\cup}}}

\newcommand{\ifnew}{\mbox{\upshape \textsf{if}}}
\newcommand{\thennew}{\mbox{\upshape \textsf{then}}}
\newcommand{\elsenew}{\mbox{\upshape \textsf{else}}}
\newcommand{\choice}{\mbox{\upshape \textsf{ch}}}
\newcommand{\arbchoice}{\mbox{\upshape \textsf{arb}}}
\newcommand{\acycchoice}{\mbox{\upshape \textsf{ac}}}
\newcommand{\cycchoice}{\mbox{\upshape \textsf{cyc}}}
\newcommand{\filter}{\ensuremath{\mathit{FIL}}}
\newcommand{\probset}{\ensuremath{\boldsymbol \Pi}}
\newcommand{\probleq}{\ensuremath{\leq_\Pi}}
\newcommand{\CommVar}{\ensuremath{\bigcap_\v} }
\newcommand{\quotfun}{\ensuremath{ \mathcal{Q}}}

\newcommand{\apre}{\mbox{\upshape \textsf{pre}}}
\newcommand{\aeff}{\mbox{\upshape \textsf{eff}}}
\newcommand{\problist}{\ensuremath \probset}
\newcommand{\cat}{{\frown}}
\newcommand{\probproj}[2]{{#1}{\downharpoonright}^{#2}}
\newcommand{\preced}{\mathbin{\rotatebox[origin=c]{180}{\ensuremath{\rhd}}}}
\newcommand{\perm}{\ensuremath{\sigma}}
\newcommand{\invstates}[1]{\ensuremath{\mathit{inv({#1})}}}
\newcommand{\probss}[1]{{\mathcal S}(#1)}
\newcommand{\parChildRel}[3]{\ensuremath{\negdep{#1}{#2}{#3}}}
\newcommand{\asessymbol}{\ensuremath{\mathbb{A}}}
\newcommand{\ases}[1]{{#1}^*}
\newcommand{\uniStates}{\ensuremath{\mathbb{U}}}
\newcommand{\recurrenceDiam}{\ensuremath{\mathit{rd}}}
\newcommand{\recurrenceAcycDiamfun}{\ensuremath{\mathit{{\mathfrak A}}}}
\newcommand{\recurrenceDiamfun}{\ensuremath{\mathit{\mathfrak R}}}
\newcommand{\traversalDiam}{\ensuremath{\mathit{td}}}
\newcommand{\traversalDiamfun}{\ensuremath{\mathit{\mathfrak T}}}
\newcommand{\isPrefix}[2]{\ensuremath{#1 \preceq #2}}
\newcommand{\aspath}{\ensuremath{\vec{\path}}}
\newcommand{\n}{\textsf{\upshape n}}
\providecommand{\graph}{}
\renewcommand{\graph}{\ensuremath{\mathcal{E}}\xspace}
\newcommand{\undirgraph}{{\cal G}}

\renewcommand{\ss}{\ensuremath{\state s}}
\newcommand{\slist}{\ensuremath{\vec{\mbox{\upshape \textsf{ss}}}}}
\newcommand{\sll}{\ensuremath{\vec{\state}}}
\newcommand{\listset}{\mbox{\upshape \textsf{set}}}
\newcommand{\asset}{\ensuremath{\mathit{K}}}
\newcommand{\aslist}{\ensuremath{\mathit{\overset{\rightarrow}{\gamma}}}}
\newcommand{\head}{\mbox{\upshape \textsf{first}}}
\renewcommand{\max}{\textsf{\upshape max}}
\newcommand{\argmax}{\textsf{\upshape argmax}}
\newcommand{\argmin}{\textsf{\upshape argmin}}
\renewcommand{\min}{\textsf{\upshape min}}
\newcommand{\bool}{\mbox{\upshape \textsf{bool}}}
\newcommand{\last}{\mbox{\upshape \textsf{last}}}
\newcommand{\front}{\mbox{\upshape \textsf{front}}}
\newcommand{\rot}{\mbox{\upshape \textsf{rot}}}
\newcommand{\stuff}{\mbox{\upshape \textsf{intlv}}}
\newcommand{\tail}{\mbox{\upshape \textsf{tail}}}
\newcommand{\ngrtoas}{\ensuremath{\mathit{\as_{\graph_\mathbb{N}}}}}
\newcommand{\vsfun}{\mbox{\upshape \textsf{vs}}}
\newcommand{\inits}{\mbox{\upshape \textsf{init}}}
\newcommand{\satprecondas}{\mbox{\upshape \textsf{sat-pre}}}
\newcommand{\remcondlessact}{\mbox{\upshape \textsf{rem-condless}}}
\providecommand{\state}{}
\renewcommand{\state}{x}
\newcommand{\statea}{x}
\newcommand{\stateb}{y}
\newcommand{\statec}{z}
\newcommand{\fals}{\mbox{\upshape \textsf{F}}}
\newcommand{\indices}{\ensuremath{V}}
\newcommand{\edges}{\ensuremath{E}}
\newcommand{\vertices}{\ensuremath{V}}
\newcommand{\listtype}{\mbox{\upshape \textsf{list}}}
\newcommand{\settype}{\mbox{\upshape \textsf{set}}}
\newcommand{\acttype}{\mbox{\upshape \textsf{action}}}
\newcommand{\graphtype}{\mbox{\upshape \textsf{graph}}}
\newcommand{\projfun}[2]{\ensuremath{\Delta_{#1}^{#2}}}
\newcommand{\snapfun}[2]{\ensuremath{\Sigma_{#1}^{#2}}}
\newcommand{\RDfun}[1]{\ensuremath{{\mathcal R}_{#1}}}
\newcommand{\elldbound}[1]{\ensuremath{{\mathcal LS}_{#1}}}
\newcommand{\distinct}{\textsf{\upshape distinct}}
\newcommand{\ddistinct}{\mbox{\upshape \textsf{ddistinct}}}
\newcommand{\simple}{\mbox{\upshape \textsf{simple}}}

\newcommand{\reachable}[3]{\ensuremath{{#1}\rightsquigarrow{#3}}}

\newcommand{\Omit}[1]{}

\newcommand{\charles}[1]{\textcolor{red}{#1}}

\newcommand{\negreachable}[3]{\ensuremath{{#2}\not\rightsquigarrow{#3}}}
\newcommand{\wdiam}[2]{{#1}^{#2}}
\newcommand{\dsnapshot}[2]{\Delta_{#1}}
\newcommand{\ellsnapshot}[2]{{\mathcal L}_{#1}}

\newcommand{\snapshotsymbol}{|\kern-.7ex\raise.08ex\hbox{\scalebox{0.7}{$\bullet$}}}
\newcommand{\snapshot}[2]{\ensuremath{\mathrel{#1\snapshotsymbol_{#2}}}}
\newcommand{\vstype}{\texttt{\upshape VS}}
\newcommand{\vtype}{{\scriptsize \ensuremath{\dom(\delta)}}}
\newcommand{\Balgo}{{\mbox{\textsc{Hyb}}}}
\newcommand{\ssgraph}[1]{\graph_\ss}
\newcommand{\agree}{\textsf{\upshape agree}}
\newcommand{\ck}{\ensuremath{\texttt{ck}}}
\newcommand{\lk}{\ensuremath{\texttt{lk}}}
\newcommand{\gr}{\ensuremath{\texttt{gr}}}
\newcommand{\gk}{\ensuremath{\texttt{gk}}}
\newcommand{\CK}{\ensuremath{\texttt{CK}}}
\newcommand{\LK}{\ensuremath{\texttt{LK}}}
\newcommand{\GR}{\ensuremath{\texttt{GR}}}
\newcommand{\GK}{\ensuremath{\texttt{GK}}}
\newcommand{\safe}{\ensuremath{\texttt{s}}}

\newcommand{\derivname}{\ensuremath{\partial}}
\newcommand{\deriv}[3]{\ensuremath{\derivname(#1,#2,#3)}}
\newcommand{\derivabbrev}[3]{\ensuremath{{\partial(#1,#2)}}}
\newcommand{\subsetoracle}{\ensuremath{ \Omega}}
\newcommand{\Aalgo}{{\mbox{\textsc{Pur}}}}
\newcommand{\Sname}{\textsf{\upshape S}}
\newcommand{\Sbrace}[1]{\Sname\langle#1\rangle}
\newcommand{\SalgoName}{\Sname_{\textsf{\upshape max}}}
\newcommand{\Salgo}[1]{\SalgoName\langle#1\rangle}

\newcommand{\WLPname}{{\mbox{\textsc{wlp}}}}
\newcommand{\WLPbrace}[1]{\WLPname\langle#1\rangle}
\newcommand{\WLPalgoName}{\WLPname_{\textsf{\upshape max}}}
\newcommand{\WLP}[1]{\WLPalgoName\langle#1\rangle}

\newcommand{\Nname}{\ensuremath{\textsf{\upshape N}}}
\newcommand{\Nbrace}[1]{\Nname\langle#1\rangle}
\newcommand{\NalgoName}{\Nname{_{\textsf{\upshape sum}}}}
\newcommand{\Nalgobrace}[1]{\NalgoName\langle#1\rangle}

\newcommand{\acycNname}{\widehat{\textsf{\upshape N}}}
\newcommand{\acycNbrace}[1]{\acycNname\langle#1\rangle}
\newcommand{\acycNalgoName}{\acycNname{_{\textsf{\upshape sum}}}}
\newcommand{\acycNalgobrace}[1]{\acycNalgoName\langle#1\rangle}

\newcommand{\Mname}{\ensuremath{\textsf{\upshape M}}}
\newcommand{\Mbrace}[1]{\Mname\langle#1\rangle}
\newcommand{\MalgoName}{\Mname{_{\textsf{\upshape sum}}}}
\newcommand{\Malgobrace}[1]{\MalgoName\langle#1\rangle}
\newcommand{\cardinality}[1]{{\ensuremath{|#1|}}}
\newcommand{\length}[1]{\cardinality{#1}}
\newcommand{\basecasefun}{\ensuremath{b}}
\newcommand{\Basecasefun}{\ensuremath{\mathcal B}}

\newcommand{\edgegen}{\ensuremath{e}}
\newcommand{\vertexgen}{\ensuremath{u}}
\newcommand{\vertexa}{{\ensuremath{\vertexgen_1}}}
\newcommand{\vertexb}{{\ensuremath{\vertexgen_2}}}
\newcommand{\vertexc}{{\ensuremath{\vertexgen_3}}}
\newcommand{\vertexd}{{\ensuremath{\vertexgen_4}}}
\newcommand{\vertexe}{{\ensuremath{\vertexgen_5}}}
\newcommand{\vertexf}{{\ensuremath{\vertexgen_6}}}
\newcommand{\vertexg}{{\ensuremath{\vertexgen_7}}}
\newcommand{\vertexsetgen}{\ensuremath{\mathit{us}}}
\newcommand{\vertexseta}{\vertexsetgen_1}
\newcommand{\vertexsetb}{\vertexsetgen_2}
\newcommand{\labelsymbol}{\ensuremath{l}}
\newcommand{\labelfun}{\ensuremath{\mathcal{L}}}
\newcommand{\DAG}{\ensuremath{A}}
\newcommand{\NalgoNameN}{{\ensuremath{\NalgoName_{\mathbb{N}}}}}
\newcommand{\NnameN}{\ensuremath{\Nname_\mathbb{N}}}
\newcommand{\replaceprojsinglename}{\raisebox{-0.3mm} {\scalebox{0.7}{\textpmhg{H}}}}
\newcommand{\replaceprojsingle}[3] {{ #2} \underset {#1} {\raisebox{-0.3mm} {\scalebox{0.7}{\textpmhg{H}}}}  #3}
\newcommand{\HOLreplaceprojsingle}[1]{\underset {#1} {\raisebox{-0.3mm} {\scalebox{0.7}{\textpmhg{H}}}}}

\newcommand{\lotus}{{\scalebox{0.6}{\includegraphics{lotus.pdf}}}}
\newcommand{\invlotus}{\mathbin{\rotatebox[origin=c]{180}{$\lotus$}}}
\newcommand{\clique}{\ensuremath{K}}
\newcommand{\partition}{\ensuremath{\vs_{1..n}}}
\newcommand{\partitiontype}{\ensuremath{\vstype_{1..n}}}
\newcommand{\vtxpartition}{\ensuremath{P}}

\newcommand{\traversalDiamAlgo}{{\mbox{\textsc{TravDiam}}}}
\newcommand{\prefix}{\textsf{\upshape pfx}}
\newcommand{\powerset}{\mathbb{P}}
\newcommand{\postfix}{\textsf{\upshape sfx}}
\newcommand{\dfunproj}{\ensuremath{{\mathfrak D}}}
\newcommand{\dfunsnap}{\ensuremath{{\textgoth D}}}
\newcommand{\ellfunproj}{\ensuremath{\mathfrak L}}
\newcommand{\ellfunsnap}{\ensuremath{\textgoth L}}
\newcommand{\cycle}{\ensuremath{C}}
\newcommand{\petal}{\ensuremath{\eta}}
\renewcommand{\prod}{\ensuremath{{{{{\mathlarger{\mathlarger {{\mathlarger {\Pi}}}}}}}}}}
\newcommand{\sccset}{{\ensuremath{SCC}}}
\newcommand{\scc}{{\ensuremath{scc}}}
\newcommand{\negate}[1]{\overline{#1}}
\newcommand{\setofsets}{\ensuremath{S}}
\newcommand{\group}{\ensuremath{\cal \Gamma}}
\newcommand{\neededvars}{{\cal N}}
\newcommand{\sspace}{\mbox{\upshape \textsf{sspc}}}
\newcommand{\tip}{\ensuremath{t}}
\newcommand{\vara}{\ensuremath{\v_1}}
\newcommand{\varb}{\ensuremath{\v_2}}
\newcommand{\varc}{\ensuremath{\v_3}}
\newcommand{\vard}{\ensuremath{\v_4}}
\newcommand{\vare}{\ensuremath{\v_5}}
\newcommand{\varf}{\ensuremath{\v_6}}
\newcommand{\varg}{\ensuremath{\v_7}}
\newcommand{\varh}{\ensuremath{\v_8}}
\newcommand{\vari}{\ensuremath{\v_9}}
\newcommand{\acta}{\ensuremath{\act_1}}
\newcommand{\actb}{\ensuremath{\act_2}}
\newcommand{\actc}{\ensuremath{\act_3}}
\newcommand{\actd}{\ensuremath{\act_4}}
\newcommand{\acte}{\ensuremath{\act_5}}
\newcommand{\actf}{\ensuremath{\act_6}}
\newcommand{\actg}{\ensuremath{\act_7}}
\newcommand{\acth}{\ensuremath{\act_8}}
\newcommand{\acti}{\ensuremath{\act_9}}

\tikzset{dots/.style args={#1per #2}{line cap=round,dash pattern=on 0 off #2/#1}}
\providecommand{\moham}[1]{\fbox{{\bf \@Mohammad: }#1}}
\newcommand{\TDbound}{{\mbox{\textsc{Arb}}}}
\newcommand{\expbound}{{\mbox{\textsc{Exp}}}}
\newcommand{\sasdom}{\expbound}
\newcommand{\cardfun}{\ensuremath{\mathbb{C}}}
\newcommand{\AGNa}{AGN1}
\newcommand{\AGNb}{AGN2}
\newcommand{\reset}{{\ensuremath{reset}}}

\newcommand{\matching}{\ensuremath{\mathcal{M}}\xspace}
\newcommand{\BlossomAlg}{{\mbox{\textsc{Find\_Max\_Matching}}}}
\newcommand{\AugPathAlg}{{\mbox{\textsc{Aug\_Path\_Search}}}}
\newcommand{\BlossomOrAugPath}{{\mbox{\textsc{Compute\_Blossom}}}}

\renewcommand{\vertices}{\ensuremath{\mathcal{V}}\xspace}
\newcommand{\lparty}{\ensuremath{L}}
\newcommand{\rparty}{\ensuremath{R}}
\newcommand{\lperm}{\ensuremath{\sigma}}
\newcommand{\rperm}{\ensuremath{\omega}}
\renewcommand{\vertexgen}{\ensuremath{v}}
\newcommand{\lvertexgen}{\ensuremath{v}}
\newcommand{\lvertexa}{\lvertexgen_1}
\newcommand{\lvertexb}{\lvertexgen_2}
\newcommand{\lvertexc}{\lvertexgen_3}
\newcommand{\lvertexd}{\lvertexgen_4}
\newcommand{\lvertexe}{\lvertexgen_5}
\newcommand{\lvertexf}{\lvertexgen_6}
\newcommand{\rvertexgen}{\ensuremath{u}}
\newcommand{\rvertexa}{\rvertexgen_1}
\newcommand{\rvertexb}{\rvertexgen_2}
\newcommand{\rvertexc}{\rvertexgen_3}
\newcommand{\rvertexd}{\rvertexgen_4}
\newcommand{\rvertexe}{\rvertexgen_5}
\newcommand{\rvertexf}{\rvertexgen_6}
\newcommand{\lorder}{\ensuremath{\pi}}
\newcommand{\rorder}{\ensuremath{\sigma}}
\newcommand{\rank}{\textit{online-match}}
\newcommand{\neighb}[2]{\ensuremath{{N_{#1} ({#2})}}}
\newcommand{\shiftsto}{\textit{shifts-to}}
\newcommand{\zig}{\textit{zig}}
\newcommand{\zag}{\textit{zag}}
\newcommand{\lpartyperm}{\ensuremath{\lparty'}}

\newcommand{\isaname}[1]{\emph{#1}}

\newcommand{\nth}[2]{#1[#2]}
\newcommand*{\uniform}{\mathcal{U}}
\newcommand*{\rankingprob}{\textit{RANKING}}
\newcommand*{\perms}{\mathcal{S}}
\newcommand*{\bernoulli}{\mathbb{I}}
\newcommand{\card}[1]{|#1|}

\section{Introduction}
The Primal-dual paradigm for analysing algorithms is one of the most successful. Its history spans more than 70 years, with the Hungarian Method \cite{KuhnHungarian} being one of the first algorithms for solving weighted bipartite matchings to follow this paradigm.
Since then, the paradigm has been used to design algorithms for problems in combinatorial optimisation, including Edmonds' algorithm for weighted matching in general graphs~\cite{EdmondsWeightedMatching}, all the way to modern analyses of online matching algorithms~\cite{DevanurOnlineMatchingPrimalDual} and some of the fastest algorithms for solving flow problems~\cite{almostLinearTimeFlows}.
In addition to exact algorithms, the primal-dual approach is also cornerstone to the design for a majority of approximation algorithms for optimisation problems, like MaxSAT, set cover, Steiner trees, etc. (cf.\ Part II of Vazirani's seminal book on approximation algorithms, which is dedicated to the primal-dual method).

In this work we present a series of formal analyses of primal-dual matching algorithms.
Our analyses cover a range of primal-dual algorithms: from the HM, arguably the first primal-dual algorithm, to more probabilistic arguments used to analyse online algorithms, including the well-known Adwords algorithm~\cite{AdWords2007} and the RANKING algorithm~\cite{KVV90}.
Our longer term goal is to create a formal library of lemmas and reasoning principles that aid in the analysis of primal-dual algorithms.

\subparagraph*{Availability.} We build on an ongoing effort to build a library of combinatorial optimisation in Isabelle/HOL~\cite{IGL}.
We plan to integrate our work in that library.

\subparagraph*{Theory Background.}
We assume the reader to be familiar with basic graph theory, such as vertices \vertices, edges \graph and paths in graphs, bipartiteness, and adjacency and incidence matrices encoding graphs.
Because it eases the formalisation, we identify a graph with its edges and define $\vertices = \bigcup \graph$. 
A \textit{matching} \matching in a graph \graph is a vertex-disjoint subset of the edges in \graph.
Matchings covering all vertices are \textit{perfect}.
We search for matchings that are optimum w.r.t.\ an optimisation objective e.g.\ the maximisation/minimisation of accumulated real weights.

In a primal-dual (PD) method, we maintain an upper bound for the value of the solution, and finally obtain a solution whose value equals the final upper bound (or is close to it), which implies optimality (or sufficient proximity).
We encode the optimisation problem into linear (in-)equalities alongside a linear optimisation objective (\textit{linear program/LP}), e.g. $\max\lbrace c^Tx.\, Ax \leq b, x \geq 0 \rbrace$ in linear algebra notation.
$A$ is usually the incidence matrix, a primal solution is a vector $x_{\matching} \in \lbrace0,1 \rbrace^{|\graph|}$ encoding a matching, and a dual solution $\pi$ is a $|\vertices|$-dimensional vector/function $\pi:\vertices \rightarrow \mathbb{R}$ (``potential''). 
We write accumulated weights as $w(E)$ for $E \subseteq \graph$, and $\pi(\vertices)$ for the potentials $\pi$.

For $\max\lbrace c^Tx.\, Ax \leq b, x \geq 0 \rbrace$, the \textit{primal}, there is a \textit{dual} $\min \lbrace b^Ty.\, A^Ty \geq c, y\geq 0 \rbrace$.
Vectors $x$ and $y$ satisfying these constraints are \textit{feasible primal} and \textit{dual solutions}, for which $c^Tx \leq b^Ty$ holds (\textit{weak duality WD}).
Furthermore, if $((1-\delta)A\hat{x}-b)^T\hat{y} = 0$ and $(1-\delta)(A^T\hat{y}-c)^T\hat{x} = 0$, then $(1-\delta)c^T\hat{x} = b^T\hat{y}$ and hence both $(1-\delta)\hat{x}$ and $\hat{y}$ are optimum solutions to the respective LPs (\textit{complementary slackness CS}).
If $\delta = 0$, $\hat{x}$ is an exact optimum, otherwise $\hat{x}$ satisfies an approximation guarantee.
For further details, e.g. other LP types and their versions of WD and CS, we refer to the literature~\cite{SchrijverLinprog,KorteVygenOptimisation}.

A PD method starts with a feasible dual solution and a primal candidate solution which together satisfy CS~\cite{PDintroduction} (scaled by $1-\delta$ for approximation).
It iteratively changes both solutions to bring the primal closer to feasibility while maintaining CS.
These steps are \textit{primal-dual adjustments (PDA)}.
When the primal solution becomes feasible, it is also optimum.

\subparagraph*{Formalisation.}
Our work here focuses on presenting aspects related to formal reasoning about primal-dual analyses.
The main design choice regarding that is the representation of LPs: we reuse matrices~\cite{JordanNormalFormAFP} that others used for verifying LP theory, e.g.\ strong duality~\cite{LinprogIsabelle} and the simplex algorithm~\cite{SimplexIsabelle}.
The mathematical core of the arguments, namely, CS and WD, is purely matrix-based and then connected to graphs by translating lemmas, which relate matchings and potentials on graphs to feasible LP solutions over matrices and vectors.
The rest of the paper focuses on the reasoning of primal-dual analyses.

Other design choices peripheral to primal-dual reasoning include our definition of matching.
We reuse matching-related Isabelle/HOL formalisations~\cite{TrustworthyGraphAlgos,BlossomAlgoIsabelle,RankingIsabelle},
which already contain many of the graph-theoretic concepts introduced above.
Another decision is the approach used to model and verify algorithms, where we reuse an approach by other authors~\cite{GraphAlgosFDS,FlowsIsabelleITP,MincostFlowLMCSArxiv,MatroidsGreedoidsIsabelle,BlossomAlgoIsabelle}.
In summary, we formalise deterministic algorithms as functional programs, i.e.\ loops as recursive functions and program states (collections of the variables a program uses) as records.
We prove \textit{invariants} (properties of the state) for the initial state, and that they are preserved by the loop iterations, which shows them for the final state.
For probabilistic algorithms, we use the existing Isabelle/HOL formalisation of the Giry monad~\cite{giryMonadIsabelle} to model and reason about them.
Algorithms have subprocedures that we assume and specify with Isabelle/HOL's locales~\cite{LocalesIsabelle}.
These are contexts that fix mathematical entities (``locale constants'') and assume their properties.
Within the context, these are available for definitions and proofs and the constants can be instantiated later.
This is a \textit{stepwise refinement}~\cite{WirthRefinement} where a program instruction is decomposed into more detailed instructions.

\section{Naive Maximum Weight Bipartite Matching\label{sec:naive}\vspace*{-0.2\baselineskip}}
\newcommand{\vmul}{\ensuremath{\texttt{*}_v}\xspace}
\newcommand{\mtrans}{\ensuremath{^T}\xspace}
\newcommand{\smul}{\ensuremath{\bullet}\xspace}
\newcommand{\nr}{\texttt{nr}\xspace}
\newcommand{\nc}{\texttt{nc}\xspace}
\newcommand{\zerov}{\texttt{0}$_v$\xspace}
\newcommand{\onev}{\texttt{1}$_v$\xspace}
\newcommand{\matchingfeasible}{\texttt{matching-feasible}\xspace}
\newcommand{\maxmatchingweakduality}{\texttt{max-matching-weak-duality}\xspace}
\newcommand{\weakdualitytheoremnonnegprimal}{\texttt{weak-duality-theorem-nonneg-primal}\xspace}
\newcommand{\maxmatchingpdoptimality}{\texttt{max-matching-pd-optimality}\xspace}
\newcommand{\zeroslackiftightmatching}{\texttt{zero-slack-if-tight-matching}\xspace}
\newcommand{\maxweightiftightmatchingcoversbads}{\texttt{max-weight-if-tight-matching-covers-bads}\xspace}
\newcommand{\graphinvar}{\texttt{graph-invar}\xspace}
\newcommand{\tttG}{\texttt{G}\xspace}
\newcommand{\uv}{\ensuremath{\lbrace u, v \rbrace}}
\newcommand{\naivematching}{\textsc{NaiveMaxWeightMatching}\xspace}
\newcommand{\nonzeros}{\textit{nonzeros}\xspace}
\newcommand{\findmatchingorwitness}{\texttt{find-matching-or-witness}\xspace}
\newcommand{\match}{\texttt{match}\xspace}
\newcommand{\witness}{\texttt{witness}\xspace}
\newcommand{\ttE}{\texttt{E}\xspace}
\newcommand{\ttX}{\texttt{X}\xspace}
\newcommand{\ttC}{\texttt{C}\xspace}
\newcommand{\ttN}{\texttt{N}\xspace}
\newcommand{\ttG}{\texttt{G}\xspace}
\newcommand{\ttw}{\texttt{w}\xspace}
\newcommand{\ttM}{\texttt{M}\xspace}
\newcommand{\ttf}{\texttt{f}\xspace}
\newcommand{\fdom}{\texttt{f-dom}\xspace}
\newcommand{\ttx}{\texttt{x}\xspace}
\newcommand{\naiveprimaldual}{\texttt{naive-primal-dual}\xspace}
\newcommand{\bads}{\texttt{bads}\xspace}
\newcommand{\findepsilon}{\texttt{find-}\ensuremath{\epsilon}\xspace}
\newcommand{\vsetiteratepmap}{\texttt{vset-iterate-pmap}\xspace}
\newcommand{\fold}{\texttt{fold}\xspace}
\newcommand{\naiveprimaldualonestep}{\texttt{naive-primal-dual-one-step}\xspace}
\newcommand{\naiveprimaldualpartialcorrectness}{\allowbreak\texttt{naive-primal-dual-partial-correctness}\xspace}
\newcommand{\initpotential}{\texttt{init-potential}\xspace}
\newcommand{\naiveprimaldualtotalcorrectness}{\allowbreak\texttt{naive-primal-dual-total-correctness}\xspace}
\newcommand{\epsilonmultiples}{\texttt{multiples-of}\xspace}
\newcommand{\feasiblemaxdual}{\texttt{feasible-max-dual}\xspace}

We search for matchings $\matching \subseteq \graph$ maximising $w(\matching)$ and derive a CS-based optimality criterion.
(Maximum) integral solutions to $\max \lbrace w^Tx_{\matching}. \, Ax_{\matching} \leq 1, x_{\matching}  \geq 0\rbrace$ encode a (max-weight) matching for $w: \graph \rightarrow \mathbb{R}^+_0$ if $A$ is the incidence matrix.
$Ax_{\matching} \leq 1$ and $x_{\matching} \geq 0$ assert that \matching is a matching.
The dual is $\min \lbrace 1^T\pi. \, A^T\pi \geq w, \pi \geq 0\rbrace$, i.e. find $\pi$ with minimum $\pi(\vertices)$ s.t. $\forall \uv \in \graph.\; \pi(u) + \pi (v) \geq w(\uv)$ and $\forall v \in \vertices.\; \pi (v) \geq 0$.
Such a $\pi$ is a \textit{feasible potential}, and $w_{\pi}$, defined as $w_{\pi}(\uv) = \pi (u) + \pi (v) - w(\uv)$ for $\uv \in \graph$, is the \textit{slack}.
For a potential $\pi$, $v$ with $\pi (v) \not = 0$ is a \textit{non-zero vertex}, and $e$ with $w_{\pi} (e) = 0$ is \textit{tight}.
Tight edges form the \textit{tight subgraph} $\graph_{\pi}$.
$\Gamma_{\pi}(X)$ are those vertices $v\in\vertices\setminus X$ where there is $\uv \in \graph_{\pi}$ with $u \in X$. 
$\Delta_{\pi}(X)$ are the tight edges connecting $X$ and $\Gamma_{\pi}(X)$.
($\Delta$ and $\Gamma$ refer to the analogous notions without considering slacks.)
We use this machinery to prove a sufficient condition for the weight-maximality of a matching:

\begin{lem}[Weight-Maximality\label{lem:maxwbpmatchcrit}]
Assume (1) \matching is a matching in a graph \graph, (2) $\pi$ is a feasible vertex potential, (3) all edges in \matching are tight, i.e.\ $\matching \subseteq \graph_{\pi}$, and (4) all vertices for which $\pi(v) \not = 0$ are matched. Then, \matching is a max-weight matching for \graph.
\end{lem}

\subparagraph*{The Algorithm.} 
Algorithm~\ref{algo:naive} works on \graph bipartite over $\lparty$ and $\rparty$ and $w:\graph \rightarrow \mathbb{R}^+_0$.
After initialising $\pi$ to a feasible potential, it tries to find a tight matching covering all non-zero vertices.
If there is such a matching, this is returned as an optimum solution.
Otherwise, it performs a \textit{primal-dual adjustment (PDA)} behind which the intuition is to move the primal (matching that matches as many non-zero vertices as possible) and the dual solution $\pi$ closer together.
If $\graph$ is bipartite and there is no matching in $\graph_{\pi}$ covering the non-zero vertices, we can find non-zero vertices $X$ with $|X| > |\Gamma_{\pi} (X)|$, as used for the PDA. 
The main invariant of this algorithm is the feasibility of $\pi$ (Invariant N1).
We have another invariant N2 saying that the vertex potentials $\pi(v)$ are integer multiples of the same real constant $\alpha$, ensuring termination for a specific class of edge weights.
We need two lemmas for correctness:

{
\begin{algorithm}[t]
\SetAlCapHSkip{0pt}
\SetAlgoHangIndent{0pt}
\SetInd{0.5em}{0.5em}
\SetVlineSkip{0.3mm}
\setlength{\algomargin}{10pt}
\caption{\naivematching(\graph, $\vertices = \lparty \cup \rparty$, $w : \graph \rightarrow \mathbb{R}_0^+$) \label{algo:naive}}
Initialise $\pi(v) = \max \lbrace w(e). \, e \in \Delta(X)\rbrace$ for $v \in \lparty$ and $\pi(v) = 0$ for $v \in \rparty$;\\
\While{$True$}
{
\textbf{if} $\exists \,\text{matching }\matching \subseteq \graph_{\pi}. \lbrace v.\, v \in \vertices \wedge \pi (v) > 0 \rbrace \subseteq \bigcup \matching$ \textbf{then return} such an \matching;\\
\Else{
find $X$ where $X\subseteq \lparty$ or $X \subseteq \rparty$ with $|X| > |\Gamma_{\pi} (X)|$ and compute \label{line:findX} \\
$\epsilon = \min (\lbrace w_{\pi}(\uv). \, u \in X \wedge v \not \in \Gamma_{\pi}(X) \wedge w_{\pi}(\uv) > 0 \rbrace \cup \lbrace \pi(v) . \, v \in X \rbrace)$;\label{line:computee}\\
\textbf{for} $x \in X$ \textbf{do} $\pi(v) \leftarrow \pi (v) - \epsilon$; \textbf{for} $x \in \Gamma_{\pi} (X)$ \textbf{do} $\pi(v) \leftarrow \pi (v) + \epsilon$;\label{line:enditeration}
}
}
\end{algorithm}
}

\begin{lem}\label{lemma:pdadjustment}
Let $\pi$ be feasible for $\graph$ and $w:\graph\rightarrow\mathbb{R}^+_0$, and $S$ be a set of vertices. 
Also, assume there is no $e \in \graph$ with $e \subseteq S$, and $\epsilon \geq 0$, $\epsilon \leq w_{\pi}(e)$ for all $e$ connecting $S$ and $\vertices \setminus S \setminus \Gamma_{\pi}(S)$, and $\forall v \in S. \;\epsilon \leq \pi(v)$. 
Let $\pi'(v) = \pi (v) - \epsilon$ for $v \in S$, $\pi'(v) = \pi (v) + \epsilon$ for $v \in \Gamma_{\pi}(S)$, and $\pi'(v) = \pi(v)$ otherwise.
$\pi'$ is feasible, and $\pi'(\vertices) = \pi(\vertices) + (|\Gamma_{\pi}(S)| - |S|)\cdot \epsilon$.
\end{lem}
\begin{proof}[Proof Ideas]
Take $\uv \in \graph$, we know $\pi(u) + \pi(v) \geq w(\uv)$, and $\pi'(u) + \pi'(v) \geq w(\uv)$ follows by a case analysis.
Also $\pi'(v) \geq 0$ because $\pi'(v) \leq \pi(v)$ only if $v \in S$, and then, $\pi'(v) = \pi(v) - \epsilon \geq 0$.
The statement on  $\pi'(\vertices)$ follows from the definition of $\pi'$.
\end{proof}
\begin{lem}\label{lemma:naivecorr}
If Algorithm~\ref{algo:naive} terminates on \graph bipartite over $\lparty$ and $\rparty$ with $w:\graph \rightarrow \mathbb{R}^+_0$, it returns a max-weight matching.
If all $w(e)$ are integer multiples of $\alpha \in \mathbb{R}$, it terminates.
\end{lem}

Algorithm~\ref{algo:naive} is a typical primal dual algorithm: 
keep an upper bound (here $\pi(\vertices)$) for the value of any solution (here $w(\matching)$ of any matching \matching).
If a solution with value equal to the upper bound cannot be found yet, we lower the upper bound carefully and continue.
If $\matching$ with $w(\matching) = \pi(\vertices)$ (inferred from CS) is found, however, $\matching$ must have maximum weight.

\section{The Hungarian Method}
\newcommand{\missed}{\textit{missed}\xspace}
\newcommand{\hungarian}{\textsc{HungarianMethod}\xspace}
\newcommand{\pathsearch}{\textsc{PathSearch}\xspace}
\newcommand{\pathsearchpar}{\textsc{PathSearch}()\xspace}
\newcommand{\infeasible}{\textit{infeasible}\xspace}
\newcommand{\hmflag}{\textit{flag}\xspace}
\newcommand{\dualunbounded}{\textit{dual-unbounded}\xspace}
\newcommand{\Lmatched}{\textit{L-matched}\xspace}
\newcommand{\nextiteration}{\textit{next-iteration}\xspace}
\newcommand{\hungarianstate}{\texttt{hungarian-state}\xspace}
\newcommand{\potentialtype}{\texttt{'potential}\xspace}
\newcommand{\matchingtype}{\texttt{'matching}\xspace}
\newcommand{\ttpathsearch}{\texttt{path-search}\xspace}
\newcommand{\hungarianloop}{\texttt{hungarian-loop}\xspace}
\newcommand{\ttlpar}{\texttt{(}\xspace}
\newcommand{\ttrpar}{\texttt{)}\xspace}
\newcommand{\tthungarian}{\texttt{hungarian}\xspace}
\newcommand{\pathsearchprecond}{\texttt{path-search-precond}\xspace}
\newcommand{\goodsearchresult}{\texttt{good-search-result}\xspace}
\newcommand{\stateinvarpresonestep}{\texttt{state-invar-pres-one-step}\xspace}
\newcommand{\cardincrease}{\texttt{card-increase}\xspace}
\newcommand{\hungariancorrectness}{\texttt{hungarian-correctness}\xspace}
\newcommand{\hungarianlooptermination}{\texttt{hungarian-loop-termination}\xspace}
\newcommand{\stateinvar}{\texttt{state-invar}\xspace}
\newcommand{\buddies}{\texttt{buddies}\xspace}
\newcommand{\statehere}{\texttt{state}\xspace}
\newcommand{\cardhere}{\texttt{card}\xspace}
\newcommand{\F}{\ensuremath{\mathcal{F}}\xspace}
\newcommand{\evens}{\textit{evens}\xspace}
\newcommand{\odds}{\textit{odds}\xspace}
\newcommand{\frees}{\textit{frees}\xspace}
\newcommand{\roots}{\textit{roots}\xspace}
\newcommand{\getpath}{\textit{get-path}\xspace}
\newcommand{\ttevens}{\texttt{evens}\xspace}
\newcommand{\ttodds}{\texttt{odds}\xspace}
\newcommand{\ttroots}{\texttt{roots}\xspace}
\newcommand{\ttgetpath}{\texttt{get-path}\xspace}
\newcommand{\abstractforest}{\texttt{abstract-forest}\xspace}
\newcommand{\ben}{\textit{ben}\xspace}
\newcommand{\heap}{\textit{heap}\xspace}
\newcommand{\hinsert}{\textsf{insert}\xspace}
\newcommand{\hdecrease}{\textsf{decrease}\xspace}
\newcommand{\hdelmin}{\textsf{del-min}\xspace}
\newcommand{\rundef}{\textsf{undef}\xspace}
\newcommand{\hempty}{\texttt{empty}\xspace}
\newcommand{\extractmin}{\texttt{extract-min}\xspace}
\newcommand{\decrkey}{\texttt{decr-key}\xspace}
\newcommand{\thinsert}{\texttt{insert}\xspace}
\newcommand{\thabstract}{\texttt{abstract}\xspace}
\newcommand{\invar}{\texttt{invar}\xspace}
\newcommand{\edgecosts}{\texttt{edge-costs}\xspace}
\newcommand{\vsetiterateben}{\texttt{vset-iterate-ben}\xspace}
\newcommand{\searchpath}{\texttt{search-path}\xspace}
\newcommand{\initialstate}{\texttt{initial-state}\xspace}
\newcommand{\searchpathloop}{\texttt{search-path-loop}\xspace}
\newcommand{\primaldualpathsearch}{\texttt{primal-dual-path-search}\xspace}

Even with termination, Algorithm~\ref{algo:naive} can have an exponential running time.
However, one can combine PDAs with the search for augmenting paths (augpaths).
For a matching \matching, these are paths $p \subseteq \graph$ such that the symmetric difference $\matching\oplus p$ is a matching $\matching'$ with $|\matching'| > |\matching|$.
This guarantees termination in polynomial time.
Although the principle of PD algorithms is simple, polynomial-time implementations and the verification thereof can be surprisingly hard.
Because optimality criterion and path search are simpler for min-weight perfect matchings, we consider this problem instead.

These matchings are integral solutions to $\min \lbrace w^Tx_{\matching}. \, Ax_{\matching} = 1, x_{\matching}  \geq 0\rbrace$,
whose dual is $\max \lbrace 1^T\pi. \, A^T\pi \leq w\rbrace$.
Feasibility of the potential (``mp-feasibility'') now means that $\pi(u) + \pi(v) \leq w(\uv)$ for all $\uv \in \graph$ \textit{without} $\forall \, v\in\vertices.\;\pi(v) \geq 0$. 
WD is here that $1^T\pi \leq w^Tx_{\matching}$ for an mp-feasible $\pi$ and perfect matching \matching.
CS is simplified to feasible $\hat{x}$ and $\hat{\pi}$ being optimum solutions for the respective LPs iff $(A^T\hat{\pi} - w)^T\hat{x} = 0$.
Similarly to Lemma~\ref{lem:minperfbpmatchcrit}, we formalised this optimality criterion:

\begin{lem}\label{lem:minperfbpmatchcrit}
If (1) \matching is a perfect matching in \graph, (2) $\pi$ is an mp-feasible vertex potential and (3) edges in \matching are tight w.r.t.\ $\pi$, then \matching is a min-weight perfect matching for \graph.
\end{lem}

There is verified executable code for max-weight bipartite matching because we reduce this and 4 other problems to min-weight perfect matching.
We extend the graph \graph to a complete bipartite graph $\graph'$ with sides of equal size ($|\lparty'| =|\rparty'|$ and $\graph' = \lbrace \uv. \; u \in \lparty' \wedge v \in \rparty' \rbrace$), ensuring the existence of a perfect matching.
If edges in $\graph' \setminus \graph$ should be avoided, we impose a penalty weight on them.
From a min-weight perfect matching for the extended weights in the new graph, we select edges forming an optimum solution for the original problem.

\begin{algorithm}[t!]
\SetAlCapHSkip{0pt}
\SetAlgoHangIndent{0pt}
\SetInd{0.5em}{0.5em}
\SetVlineSkip{0.3mm}
\setlength{\algomargin}{10pt}
\SetKwFor{While}{while}{}{}
\caption{\hungarian(\graph, $\vertices = \lparty \cup \rparty$, $w : \graph \rightarrow \mathbb{R}, \textit{initial mp-feasible } \pi$) \label{algo:hungarian}}
\textbf{if} $|\lparty| \not = |\rparty|$ \textbf{return} $\infeasible$; \textbf{else} initialise $\matching = \emptyset$;\\
\While{$True$ {\upshape \textbf{do} $(\hmflag, \pi', p) \leftarrow \pathsearch(\graph, \lparty, \rparty, w, \matching,\pi)$;}}
{
\textbf{if} $\hmflag = \dualunbounded$ \textbf{then return} $\infeasible$\\
\textbf{else if} $\hmflag = \Lmatched$ \textbf{then return} $\matching$
\textbf{else} [$\matching \leftarrow \matching \oplus p$; $\pi \leftarrow \pi'$;]
}
\end{algorithm}

\subparagraph*{Top Loop.}
The Hungarian Method (HM)~\cite{Jacobi,KuhnHungarian,KuhnHungarianVariants} (Algorithm~\ref{algo:hungarian}) expects a bipartite graph given by $\graph$, $\lparty$ and $\rparty$, a weight function $w$ and an initially feasible potential $\pi$.
After checking for obvious infeasibility if $|\lparty|\not = |\rparty|$ and initialising \matching as $\emptyset$, the main loop repeatedly applies \pathsearchpar which returns a status flag.
If \hmflag says that the dual LP is unbounded, infeasibility follows.
If $\lparty$ is matched, the current matching is returned as min-weight perfect matching for \graph.
Otherwise, \pathsearch also returns a new feasible potential $\pi'$ and an \matching -augpath $p$.
$p$ is used to augment \matching and the algorithm continues with the next iteration.

Algorithm~\ref{algo:hungarian} maintains as invariants that (HM1) \matching is a matching in \graph, (HM2) $\matching \subseteq \graph_{\pi}$, and (HM3) $\pi$ is mp-feasible w.r.t.\ $w$ and \graph. 
Provided that \matching is a matching in $\graph_{\pi}$ and $\pi$ is mp-feasible, we assume three properties for $(\hmflag, \pi', p) = \pathsearch(\graph, \lparty, \rparty, w, \matching,\pi)$:
(P1) If $\hmflag = \dualunbounded$ there is an mp-feasible $\pi'$ with $\pi'(\vertices) > B$ for any $B \in \mathbb{R}$. 
(P2) If $\hmflag = \Lmatched$, $\lparty \subseteq \bigcup \matching$.
(P3) If $\hmflag = \nextiteration$, $\pi'$ is mp-feasible, $\matching \subseteq \graph_{\pi'}$, $p \subseteq \graph_{\pi'}$ and $p$ is an $\matching$-augpath.
Assuming P1-P3, the following correctness lemma holds:

\begin{lem}
Let \graph be bipartite over $\lparty$ and $\rparty$, and $\pi$ be mp-feasible w.r.t.\ \graph and $w:\graph\rightarrow\mathbb{R}$.
On this input, the HM terminates, and returns a min-weight perfect matching if there is one.
\end{lem}

\subparagraph*{Path Search.}
We search for augpaths $p \subseteq \graph_{\pi}$ that are tight w.r.t.\ an mp-feasible $\pi$, for which \textit{alternating forests (AF)} are the central tool~\cite{BlossomAlgo}.
Other authors formalised a non-executable version for the blossom algorithm~\cite{TrustworthyGraphAlgos,BlossomAlgoIsabelle},
and we provide an executable one that works well with forest growth and PDAs happening at the same time.
An $\epsilon$ is determined similar to Algorithm~\ref{algo:naive} and preservation of mp-feasibility is similar to Lemma~\ref{lemma:pdadjustment}.
Efficient search requires ``caching'' of slacks and a priority queue.
We used \textit{priority search trees}~\cite{LammichNipkowDijkstraPrim} for the queue and red-black trees for other data structures, resulting in a verified $\mathcal{O}(n\cdot(n+m) \cdot \log n)$ implementation of the HM, which is the best possible running time of a purely functional implementation.
The verification effort for \pathsearch was high with 8 involved invariants.

\section{Online Matching}
\newcommand{\expect}{\ensuremath{\mathbb{E}}}
Online bipartite matching is a variant of bipartite matching where vertices in one party of the graph arrive online, one-by-one, each along with its incident edges.
After a vertex' arrival, the algorithm irrevocably decides on whether any of its incident edges is included in the matching.
This setting was introduced by Karp et al.~\cite{KVV90}, where they considered unweighted bipartite matching.
The online model of matching has recently gained substantial interest~\cite{rankingHighProbability,onlinebmatching,fullyOnlineMatchingI,fullyOnlineMatchingII} due to the proliferation of the internet and different applications which could be modelled and understood as online matching problems.
Most notable among those is the Adwords~\cite{AdWords2007} algorithm, which is an idealised model of how advertisers bid on keywords to show their ads to users who search for those keywords.

RANKING~\cite{KVV90}, Adwords~\cite{AdWords2007}, and most other online matching algorithms have proofs that are quite complex to understand, let alone formalise.
Often, this is primarily due to the complex combinatorial arguments used to show them correct.
For instance, the correctness of RANKING's initial analysis by Karp et al., which was formalised earlier by Abdulaziz and Madlener~\cite{RankingIsabelle}, was improved over 6 times by a number of authors~\cite{goelOnlineBudgetedMatching2008,onlineMatchingSimple2008,DevanurOnlineMatchingPrimalDual,onlineMatchingEcon,vaziraniOnlineMatching2022,rankingHighProbability}, similarly to improvements to the analysis~\cite{DevanurOnlineMatchingPrimalDual,vaziraniOnlineMatchingArxiv} of Adwords.
In this work, as part of our framework, we formalise a PD-analysis that has substantially simplified the competitiveness analysis of a number of online matching algorithms.
We briefly describe those formalisations.

\begin{algorithm}[t!]
\SetAlCapHSkip{0pt}
\SetAlgoHangIndent{0pt}
\SetInd{0.5em}{0.5em}
\SetVlineSkip{0.3mm}
\setlength{\algomargin}{10pt}
\SetKwFor{While}{while}{}{}
\caption{$\rankingprob(\graph \text{ bipartite over } L \text{ and } R, \text{arrival order } \rperm \text{ for } \rparty)$ \label{algo:ranking}}
$\lperm \gets$ a random permutation of $\lparty$; $\matching\gets\emptyset$;\\
\For{every arriving vertex $\rvertexgen$ in $\rperm$}{
\lIf{$\exists\lvertexgen\in \gamma(\rvertexgen) \setminus \bigcup\matching$}{$\matching\gets \matching \cup \{\{\argmin_{\lvertexgen\in \gamma(\rvertexgen) \setminus \bigcup\matching} \lperm(\lvertexgen),\rvertexgen\}\}$}
}
\Return $\matching$;
\end{algorithm}

We first discuss the RANKING algorithm briefly, following the previous formalisation~\cite{RankingIsabelle}.
In its original form, RANKING operates as sketched in Algorithm~\ref{algo:ranking}:
it takes as an input a bipartite graph with edges $\graph$, and a list $\rperm$ with the order of arrival of the online vertices $\rparty$.
The offline party $\lparty$ is permuted uniformly randomly, leading to a list $\lperm$.
We then go through the online vertices in the order given by $\rperm$, adding to the matching an edge that, if it exists, connects the incoming vertex $u$ to the unmatched offline vertex $v$ with the highest ranking w.r.t.\ $\lperm$. The neighbourhood, i.e. all vertices a vertex $u$ is connected to, is $\gamma(u)$.
The main theorem to prove about RANKING pertains to its competitive ratio:
\begin{theorem}
The competitive ratio of RANKING for an instance with a maximum cardinality matching of size $n$ is at least $1 - \frac{1}{e}$, i.e. 
$1 - \frac{1}{e} \leq \frac{\expect_{R \sim \text{RANKING}(G,\pi)}[R]}{n}$.
\end{theorem}

The main idea of the primal-dual analysis of RANKING is as follows: we have a primal LP relaxation of bipartite matching $\max \lbrace x_{\matching}. \, Ax_{\matching} \leq 1, x_{\matching}  \geq 0\rbrace$ and its dual $\min \lbrace 1^T\pi. \, A^T\pi \geq 1, \pi \geq 0\rbrace$, i.e.\ find $\pi$ with minimum $\pi(\vertices)$ s.t.\ $\forall u\in\lparty, v\in\rparty, \uv \in \graph.\; \pi(u) + \pi (v) \geq 1$ and $\forall v \in \vertices.\; \pi (v) \geq 0$.
Now, as mentioned earlier, a primal-dual analysis of an optimisation algorithm usually proceeds by proving a relationship between the primal objective and the dual objective.
Proving this relation for RANKING is challenging as we need to tackle the difficulty with randomisation in RANKING, which makes it a non-standard primal-dual analysis.
We can only reason about the expectation of the algorithm's output and thus we can only reason about the relationship between expected values of the primal and dual LPs.

Both LPs have to be related to or derived from the output, as in case of the two offline algorithms discussed earlier.
The primal LP still corresponds to the computed matching \matching, but now the main challenge is connecting the dual LP to \matching: a different approach specific to the online setting is needed.
In particular, the potential has to include information about the ordering of the offline vertices after they are permuted.
To reason about the dual LP's expected objective, we replace the permutation step with choosing a real-numbered priority $Y_i \in [0, 1]$ for each offline vertex $i$.
This is to allow for the expectation to be over an integrable function.
Then each constraint in the dual LP is instantiated as follows: $\pi(u)= g(Y_i)/F$ and $\pi(v) = (1 - g(Y_i))/F$ for $g: [0, 1] \to [0, 1]$ that is monotone and $0 < F \le 1$, for every $u\in\lparty$, $v\in\rparty$, and $\{u,v\}\in\graph$.
This allows us to formalise the following reasoning:
\begin{align*}
    \expect[x_\matching] &= \expect[1^T \cdot (\frac{1}{F}\cdot \pi)] \quad \text{(by construction of the primal and dual LPs)} \\
    &= (1^T \cdot \expect[\pi]) /F\quad \text{(linearity of expectation)} \\
    &\ge 1 \cdot x^* \cdot F \quad \text{(by weak duality, where } x^* \text{ is an optimal primal solution)} \\
    &= 1 \cdot |\matching^*| \cdot F \quad \text{(by definition of the primal LP)}
\end{align*}
for any max-card matching $\matching^*$. The last step is choosing $g$ and $F$, where we set $F = 1 - 1/e$, thus proving the bound.\footnote{The expectations are taken over the distribution of priorities $Y$. A step we gloss over here is the equivalence of expectations over permutations and priorities.}

This approach has two advantages:
it is much simpler than the combinatorial approach as formalised by Abdulaziz and Madlener as the primal-dual-based proof is less than one half (3K lines) of the combinatorial one.
Moreover, it provides a general framework for reasoning about online matching algorithms.
We used our development to (a) formalise the competitive analysis for the vertex-weighted variant of online matching~\cite{onlineVertexWeightedMatching}, where the primal and dual LPs are adapted to include the weights with the entire analysis remaining unchanged.
(b) With some changes, we did a primal-dual analysis of the Adwords algorithm~\cite{AdWords2007}, which solves a variant of online b-matching, that models the assignment of keywords to advertisers in the context of search engines.
Interested readers should refer to the formalisations.

\section{Discussion}
There is a rich literature on the formal analysis of algorithms, including approximation algorithms~\cite{nipkowApprox}, matching algorithms~\cite{BlossomAlgoIsabelle}, flow algorithms~\cite{FlowsIsabelleITP,LammichFlows}.
Our work here, however, addresses a large gap in the literature, namely, the formalisation of primal-dual analyses of algorithms.
We primarily focus on matching algorithms here, covering classical and modern results.
Our formalisations, which are around 14K lines, cover a large variety of reasoning styles: the Hungarian Method, which is an executable practical algorithm, and variants of online matching algorithms, which are primarily used to theoretically analyse online markets rather than to solve practical problems. 

PD-based reasoning about the correctness of algorithms is mostly algebraic and thus leads to simpler, shorter and more textbook-style proofs.
The formalised combinatorial arguments for the correctness of RANKING~\cite{RankingIsabelle} and Berge's Lemma for the blossom algorithm~\cite{BlossomAlgoIsabelle} are of much higher complexity, come with extensive case analyses and are harder to understand than the PD-arguments for the HM and RANKING, for example.
We believe PD-based reasoning can also simplify correctness arguments for minimum cost flows, for which there was only a combinatorial proof formalised so far~\cite{FlowsIsabelleITP}.
According to standard literature~\cite{KorteVygenOptimisation,schrijverBook}, maintenance of dual variables even leads to faster algorithms for this problem.

However, a big missing part in our work that we aim to contribute to a library~\cite{IGL} is the analysis of primal-dual approximation algorithms: these indeed form the majority of applications of the primal-dual paradigm in theoretical computer science.
Our immediate plans are to formalise algorithms for approximating MaxSAT, set cover, load balancing, and Steiner trees, all of which are milestones in the theory of approximation algorithms.

\bibliography{bpmatching,long_paper}

\end{document}